\newif\ifllncs
\ifllncs
  \documentclass[envcountsame,envcountsect]{llncs}
\else
  \documentclass[11pt]{article}
  \usepackage[margin=1in]{geometry}
  \usepackage{amsthm}
\fi

\newif\ifshowtoc
\ifllncs\showtocfalse\else\showtoctrue\fi

\usepackage{amsmath,amssymb,mathtools}
\usepackage{algpseudocode}
\usepackage{enumitem}
\usepackage{microtype}
\usepackage[hidelinks]{hyperref}
\usepackage[dvipsnames]{xcolor}
\usepackage{mdframed}
\usepackage{aliascnt}
\usepackage[capitalize,noabbrev,nameinlink]{cleveref}

\ifllncs\else
  \newtheorem{theorem}{Theorem}[section]
  \newaliascnt{lemma}{theorem}
  \newtheorem{lemma}[lemma]{Lemma}
  \aliascntresetthe{lemma}
  \newaliascnt{corollary}{theorem}
  \newtheorem{corollary}[corollary]{Corollary}
  \aliascntresetthe{corollary}
  \theoremstyle{definition}
  \newaliascnt{definition}{theorem}
  \newtheorem{definition}[definition]{Definition}
  \aliascntresetthe{definition}
  \newaliascnt{remark}{theorem}
  \newtheorem{remark}[remark]{Remark}
  \aliascntresetthe{remark}
\fi
\newaliascnt{assumption}{theorem}
\newtheorem{assumption}[assumption]{Assumption}
\aliascntresetthe{assumption}

\crefname{theorem}{Theorem}{Theorems}
\crefname{lemma}{Lemma}{Lemmas}
\crefname{corollary}{Corollary}{Corollaries}
\crefname{definition}{Definition}{Definitions}
\crefname{assumption}{Assumption}{Assumptions}
\crefname{remark}{Remark}{Remarks}
\crefname{equation}{Equation}{Equations}
\crefname{section}{Section}{Sections}
\Crefname{theorem}{Theorem}{Theorems}
\Crefname{lemma}{Lemma}{Lemmas}
\Crefname{corollary}{Corollary}{Corollaries}
\Crefname{definition}{Definition}{Definitions}
\Crefname{assumption}{Assumption}{Assumptions}
\Crefname{remark}{Remark}{Remarks}
\Crefname{equation}{Equation}{Equations}
\Crefname{section}{Section}{Sections}

\newcommand{\players}{\mathcal{P}}
\newcommand{\honest}{\mathsf{honest}}
\newcommand{\rational}{\mathsf{rational}}
\newcommand{\Attack}{\mathsf{Attack}}
\newcommand{\acc}{\mathsf{acc}}
\newcommand{\rej}{\mathsf{rej}}

\newcommand{\ind}{\mathbf{1}}

\newcommand{\cnote}[1]{{\color{brown}{Chen-Da: #1}}}
\newcommand{\hao}[1]{{\color{blue}{Hao: #1}}}
\renewcommand{\cnote}[1]{}
\renewcommand{\hao}[1]{}

\title{Too Late to Slash: Coordinating a Risk-Free Equivocation Attack}
\ifllncs
\author{}
\institute{}
\else
  \author{
    Hao Chung\\
    {\small Layerzero Labs}\\
    {\small\texttt{hao.chung@layerzerolabs.org}}
    \and
    Chen-Da Liu-Zhang\\
    {\small Lucerne University of Applied Sciences and Arts}\\
    {\small\texttt{chen-da.liuzhang@hslu.ch}}
  }
  \date{}
\fi

\begin{document}

\maketitle


\begin{abstract}
\emph{Slashing} is commonly argued to secure proof-of-stake blockchains
by confiscating the stake of misbehaving validators.
The usual justification is that, without slashing,
validators can solicit an equivocation attack by signing conflicting blocks:
if enough others join, the attack succeeds;
otherwise, the attempt incurs no loss.
Slashing is intended to make such attempts costly and thereby guarantee the security of applications whose economic value is comparable to the bonded stake.
This rationale, however, rests on heuristic arguments
rather than a formal game-theoretic guarantee.

We challenge this rationale by constructing a risk-free coordination protocol
for rational validators under algorithmic slashing.
We show that the prescribed strategy profile,
in which rational validators solicit other validators to equivocate,
constitutes an ex post Nash equilibrium,
even when validators do not know in advance how many others will participate.
The equilibrium holds for any gain $\epsilon>0$ from successful equivocation,
however small relative to the bonded stake.
Thus, slashing alone does not guarantee economic security
proportional to the value of bonded stake.

Our results expose a fundamental limitation of algorithmic slashing.
Whereas conventional collateral arrangements in many real-world scenarios can rely on external enforcement,
for example through courts,
algorithmic slashing depends on the same consensus process that the attackers control.
These findings call for a formal analysis of slashing's security,
rather than overly simplistic arguments drawn from financial systems
with independent enforcement.
\end{abstract}

\ifshowtoc
  \newpage
  \tableofcontents
  \newpage
\fi

%
%

\section{Introduction}

A blockchain is maintained by a committee of validators running a consensus
protocol, so that everyone who watches it sees the same ever-growing history.
Its guarantees hold as long as the total weight of faulty validators stays below
a fixed threshold. In proof of stake that weight is bonded stake, so the
guarantees hinge on how much stake is in faulty hands.

Real validators are not simply honest or faulty. They deviate when deviating
pays. The standard remedy is \emph{slashing}: a validator caught misbehaving
forfeits part or all of its bond. Slashing comes in two kinds
\cite{kannan2023slashing}. \emph{Algorithmic} slashing turns objectively
verifiable evidence into a protocol-defined penalty, with no human in the loop.
\emph{Social} slashing relies on authority outside the protocol: governance
votes, a recovery fork, a court. This paper is about the first kind, and the
difference will turn out to matter a great deal.

How is slashing supposed to work? The textbook argument goes like this. A
validator who wants to double-spend cannot do it alone; it needs enough other
stake to join. Without slashing, finding out whether it will is free: sign two
conflicting blocks and see who follows. If enough follow, the attack succeeds.
If not, nothing happens and nothing is lost. Slashing kills this free option.
The two signatures are public evidence against their author, and the stake is
gone.\footnote{Ethereum's documentation puts the cost at the staked supply times the token price \cite{ethereum_pos_vs_pow,eth2book-staking}. With roughly a hundred billion dollars staked \cite{coinbase_eth_staking}, Ethereum is said to carry tens of billions in dependent activity.} To learn whether others would join, a validator must first convict itself. This is how slashing is said to solve the nothing-at-stake problem
\cite{buterin2017casper,xiao2020survey}.

At heart this is a prediction about behavior: rational validators will not go
looking for accomplices, because looking is punishable. More precisely, it is a
claim about a \emph{coordination game}, in which validators try to discover one
another's willingness to attack: that discovery costs stake. If discovery is
free, slashing has nothing to punish until the coalition already exists.

\begin{quote}
\emph{Can rational validators find out whether enough others are willing to attack? Moreover, are they willing to try to find out even if the bonded stake is large and the attack-reward is small?}
\end{quote}


\subsection{Our Contribution}

The answer is yes, on both counts. We exhibit an attack in which rational
validators are incentivized to find out whether enough others are willing to join, irrespective of the attack-reward and the stake bond sizes. We can therefore conclude that once coordination is taken into account, the security of an application on the chain is not proportional to the stake behind it.

The reason such an attack exists is that, under algorithmic slashing, the party
that enforces the penalty and the party that is slashed are the same. A slashing
penalty is a transaction and it must be included in a block that must be finalized by the very validators the penalty is aimed at. Collateral works
in the ordinary economy because the bank, not the borrower, seizes it. Here the enforcer is
the defendant. So the punishment is credible in precisely the regime where it is
not needed: before a coalition capable of attacking exists there is nothing to
slash, and once it exists, it controls the process that would do the slashing.
It is \textbf{too late to slash}. This is also the difference from social
slashing, which brings in an enforcer the attackers do not control.


\paragraph{A model of blockchains with slashing.}
We model algorithmic slashing by making its dependence on consensus explicit
(\cref{sec:slashing-model}):
a penalty takes effect only when its transaction executes in a finalized block.
We define the \emph{monopoly threshold}
(\cref{def:monopoly-thresholds}) as the minimum coalition weight
sufficient for a single strategy to both finalize any pair of individually valid,
conflicting blocks and selectively include or exclude transactions.
This threshold captures the power needed to equivocate
while preventing slashing penalties from taking effect.
The coalition uses only ordinary protocol actions,
without breaking cryptographic primitives.

\paragraph{Risk-free coordination.}
Our first result focuses on the \emph{standard double-signing rule}:
a validator is subject to slashing if it signs two conflicting blocks
in the consensus.
The pair of signatures serves as publicly verifiable evidence of the violation.

Our strategy separates solicitation from equivocation.
Validators first register their willingness to attack through a public smart contract
(\cref{sec:coordination}).
Registration is risk-free:
validators sign only a statement of intent,
which provides no evidence of double signing.
After a public deadline, the contract checks whether the enrolled weight
reaches the monopoly threshold.
Below the threshold, coordination aborts without producing slashing evidence;
at or above it, participants equivocate and censor slashing transactions
until they withdraw their full stake.
Thus, evidence appears only once the coalition has enough power to block punishment.

\begin{theorem}[Risk-free attack equilibrium; informal, see
\cref{thm:risk-free}]
Under the standard double-signing rule,
the profile in which every rational validator follows the strategy above
forms a Nash equilibrium,
even if validators do not know in advance which others are willing to participate.
\end{theorem}

Notably, the equilibrium holds for any positive gain from successful equivocation,
however small relative to the bonded stake. 
Moreover, we also show that honest play is not always a best response. 
Each rational validator prefers joining rather than remaining honest,
provided the other rational validators follow the coordination strategy
(\cref{cor:join-incentive}). Joining yields a positive gain without losing stake,
whereas remaining honest yields no gain.
 
\paragraph{Robustness to evidence rewards.}
A natural objection is that the baseline attack works only because
no one is rewarded for reporting equivocation.
Offering such rewards could unravel the coalition from within.
In \cref{sec:evidence-rewards},
we require each participant to post a deposit
that is forfeited whenever any coalition member is slashed.
If the deposit covers the total rewards its owner could collect by reporting others,
betrayal yields no net gain.
This preserves both the attack equilibrium
(\cref{thm:reward-equilibrium})
and the strict incentive to join
(\cref{cor:reward-join-incentive}).

\paragraph{Beyond the standard rule.}
A broader slashing rule might punish deploying a coordination contract
or registering an intention to equivocate.
In \cref{sec:anon-collateral},
we consider any slashing rule that never accepts evidence against an honest validator.

We construct a coordination protocol with anonymous enrollment.
The available information is therefore consistent with a target validator
being honest.
Attributing enrollment to that validator could then implicate an honest validator,
violating no slashing against honest stake.
With sufficient collateral to make duplicate enrollment unprofitable,
we obtain a risk-free attack equilibrium 
(\cref{thm:anonymous-deposit-equilibrium}) 
and the strict incentive to join
(\cref{cor:anonymous-join-incentive}).
 

\paragraph{Impact of our results.}
Our results challenge the security rationale for algorithmic slashing:
rational validators can solicit an attack without risking their stake,
and successful equivocation can remain profitable
however small its gain relative to the bonded stake.
Slashing cannot deter sophisticated attackers who can block their own punishment.

Our aim is not to develop a deployable attack.
Our construction is idealized and requires a coalition at the monopoly threshold,
with greater consensus power than needed to violate safety alone.
In practice, slashing may still be useful for honest but careless participants:
operational errors, such as running multiple validator clients with the same key,
are a common cause of slashing \cite{saintolive2024slashing}.
These roles, however, do not establish security against sophisticated attacks
by coordinated rational validators.

Finally, our scope is algorithmic slashing.
Social slashing draws on external authority
and may provide the only effective recourse after such an attack.
That protection comes from enforcement outside the consensus protocol.
Claims of economic security must therefore rest on a formal analysis
of incentives and enforcement,
rather than on the amount of stake nominally at risk.
 
\subsection{Related Work}

\paragraph{Rational participants and accountability.}
The inadequacy of a purely honest-versus-faulty taxonomy is long recognized. The
BAR model \cite{aiyer2005bar} and work at the interface of distributed computing
and game theory \cite{abraham2006distributed,halpern2008beyond} study protocols
that are equilibria rather than merely fault tolerant. 
Ranchal-Pedrosa and Gramoli \cite{ranchal2022trap}
show that consensus can exceed the classical one-third fault threshold
when some potentially faulty participants are rational rather than Byzantine.
Our validators are rational
in the same sense, but the question differs. That line of work asks whether
rational parties would deviate from a given protocol. We ask whether they would
first coordinate: whether they are incentivized to find out that enough of them
are willing to deviate before any of them does.

A separate line of work asks who is provably at fault. Casper FFG
\cite{buterin2017casper} introduced accountable safety; later work extracts the
maximal identifiable culprit set from conflicting executions
\cite{sheng2021bft,civit2021polygraph} and characterizes the tradeoff against
availability \cite{neu2021availability}. These works assume a Byzantine
adversary, so the question of whether a coalition would ever assemble does not
arise, and they stop at identification: whether the identified fault can then
be punished is not asked.

\paragraph{Economic security of consensus.}
Budish, Lewis-Pye, and Roughgarden \cite{budish2024limits}
study the fault-tolerance thresholds of consensus protocols
and the slashing guarantees attainable under different network and permission models.
Their impossibility result under partial synchrony identifies a related limitation:
an adversary can exploit network delays to violate consistency
and withdraw its stake before punishment.
Their analysis starts with an already formed adversarial coalition
that knows the total consensus weight it controls.
We instead study how rational validators coordinate to form an attacking coalition,
without knowing in advance who will participate.
Moreover, our attack does not rely on a network partition.

Kannan and Deb \cite{kannan2023slashing} argue that slashing raises
the cost of a bribery attack to at least one third of the total stake
in their BFT model.
This bound assumes that double-signing validators lose their deposits
on a canonical fork.
Deb, Raynor, and Kannan \cite{deb2024stakesure}
develop this framework further in STAKESURE,
refining bounds on attack profits and using slashed stake to compensate insured users.
Both works also consider \emph{token toxicity}:
a successful attack may reduce the value of the staked token
and thereby impose losses on the attackers.
Our model does not capture token toxicity as a potential deterrent, 
and we examine algorithmic enforcement:
a coalition reaching the monopoly threshold can censor slashing transactions
and withdraw its stake.

\paragraph{Bribery and coordination.}
Coordinating an attack and bribery through a smart contract is not new. 
Bribery attacks on proof-of-work were formalized in \cite{bonneau2016buy,mccorry2018bribing}, and
the ``Pay-to-Win'' framework \cite{judmayer2021paytowin} builds crowdfunded
contracts that pay participants only if the attack succeeds. Our mechanism
differs in three ways: it transfers no funds, it targets proof of stake
\emph{with slashing}, and it does not buy willingness but \emph{reveals}
willingness that is already there.

Karakostas, Kiayias, and Zacharias
\cite{bribing2024counterincentives}
analyze bribery attack in proof-of-stake blockchains, and they show that slashing
can mitigate certain bribery attacks.
In their model, a validator's stake is guaranteed to be slashed if a provable evidence is presented.
Our model makes the execution of that penalty part of the game:
even when valid evidence exists.

Yeo and Zhang \cite{yeo2025censorship}
study how rational validators coordinate an attack
through a public bulletin board in the form of a smart contract to censor transactions selectively.
Our work focuses instead on formalizing whether slashing can be useful in preventing attack coordination. As such, our game-theoretic formalizations are tailored to a model for blockchain with slashing, and our attack aims to allow rational parties to coordinate in producing a double-signing evidence to profit from double spending without being punished.


Zhang et al. \cite{zhang2023collusion} show that an external bribery contract, funded by an entity denoted the \emph{magnate} on a separate chain, can make honest parties weakly dominated and collude in a strict Nash equilibrium in any incentive-based blockchain. We instead coordinate on the attacked chain itself without any bribe, and formalize when a coalition can equivocate without ever being slashed.


\section{Modeling Blockchains with Slashing}
\label{sec:slashing-model}

\subsection{Blockchains with Slashing}
We consider a proof-of-stake blockchain protocol
$\Pi_{\mathsf{bc}}$ run by a committee $\players=\{1,\ldots,n\}$.
Let $\mathsf{acct}_i=(\mathsf{pk}_i,w_i,s_i,\mathsf{aux}_i)$ be
validator $i$'s account state containing its public
key $\mathsf{pk}_i$, consensus weight $w_i>0$, stake balance $s_i\geq 0$
subject to slashing, and auxiliary information
$\mathsf{aux}_i$.\footnote{The auxiliary information includes any metadata,
such as replay-protection data used later.}
We assume $s_i>0$ at the beginning of the modeled episode. 
Let $\mathcal V_{\mathsf{bc}}=\{\mathsf{acct}_i\}_{i\in\players}$ be the
public registry with normalized weights $\sum_{i\in\players}w_i=1$.  
We assume $\mathcal V_{\mathsf{bc}}=\{\mathsf{acct}_i\}_{i\in\players}$ is a public observable information
stored on the blockchain.
Although $w_i$ is often proportional to $s_i$ in proof-of-stake protocols, 
our analysis does not require this relation.

For simplicity, we assume that the committee $\players$ and its associated
consensus weights remain fixed throughout the \emph{episode},
i.e., a sequence of blocks in which the committee is updated only at episode boundaries.
Stake balances may change within an episode, with
the corresponding changes in consensus weight taking effect in the next
episode.  The later coordination game ends after determining whether the
attack succeeds and which slashing actions execute.  
Then, utilities are evaluated.
Thus, removing a slashed validator immediately or at the next episode
boundary does not affect any game-theoretic analysis in that episode.


We leave the network model abstract.  The definitions may be instantiated in
synchronous, partially synchronous, or asynchronous settings.
Our results apply to any such protocol instantiation satisfying the stated
properties, although the resulting colluding thresholds may depend on the network model.


For a block height $h$, let $\mathsf{ctx}_h$ denote any public protocol context,
including the committee information, and previously finalized state up to but not including height $h$.
A symmetric relation $\perp_h$ identifies conflicts: $B\perp_h B'$ means
that $B$ and $B'$ are each individually valid relative to $\mathsf{ctx}_h$
but cannot both occur in one valid blockchain history.

We model finality through the deterministic public algorithm
\[
  \mathsf{VfyFinal}
    (\mathsf{ctx}_h,B,\varphi)\in\{\acc,\rej\}
\]
which verifies whether $\varphi$ is a valid finality certificate for $B$
under $\mathsf{ctx}_h$.  We say that an execution finalizes $B$ relative to
$\mathsf{ctx}_h$ if it produces a certificate $\varphi$ for which
$\mathsf{VfyFinal}(\mathsf{ctx}_h,B,\varphi)=\acc$.  In BFT-style consensus,
$\varphi$ may be a quorum certificate comprising the committee members' signatures.
Thus conflicting finality is publicly verifiable even when observers receive the two certificates at different times, and it remains meaningful
when every validator is corrupted.

Because the equivocation attacks considered here may produce conflicting
finalized histories, the blockchain state need not be globally unique.
Accordingly, blockchain states, validator accounts and balances, and
statements that a validator is slashed are always relative to a specified
finalized blockchain history.  
A validator may therefore be slashed in one finalized history but not in another.


\begin{definition}[slashing mechanism]
A slashing mechanism for $\Pi_{\mathsf{bc}}$ is a tuple
  $\bigl(\mathsf{GenEv},\mathsf{VfyEv}, \mathsf{ApplySlash}\bigr)$.
\begin{itemize}
\item 
$\mathsf{GenEv}$ is a possibly interactive protocol among validators that,
given a blockchain context $\mathsf{ctx}_h$, may output
$(\mathsf{ctx}_h,i,\mathsf{ev})$, where $\mathsf{ev}$ is an
evidence of validator $i$'s misbehavior with respect to $\mathsf{ctx}_h$.
\item 
$\mathsf{VfyEv}(\mathsf{ctx}_h,i,\mathsf{ev})$ is a deterministic algorithm
that checks whether $\mathsf{ev}$ demonstrates misbehavior by validator $i$.
It outputs $\acc$ if verification succeeds and $\rej$ otherwise.
\item 
$\mathsf{ApplySlash}(\mathsf{ctx}_h,i,\mathsf{ev})$ deterministically computes
the following transition on validator $i$'s current account:
\[
  \mathsf{ApplySlash}(\mathsf{ctx}_h,i,\mathsf{ev}):
  (\mathsf{pk}_i,w_i,s_i,\mathsf{aux}_i)
  \longmapsto
  (\mathsf{pk}_i,w_i,s'_i,\mathsf{aux}'_i),
\]
such that $s'_i\geq 0$ and $s'_i\leq s_i$.  
If $\mathsf{VfyEv}(\mathsf{ctx}_h,i,\mathsf{ev})=\rej$, or if it accepts
evidence for a violation that has already been processed, then the transition
leaves the account unchanged.
If $s'_i<s_i$, then $s_i-s'_i$ is the amount slashed.  
The transition uses the auxiliary field to recognize previously
processed violations, including when the same violation is presented through
a different encoding accepted by $\mathsf{VfyEv}$,
which is critical to avoid double slashing for the same violation.
\end{itemize}
\end{definition}

$\mathsf{ApplySlash}$ only computes an account transition.  
The transition becomes effective only when a transaction invoking it is executed
as \emph{part of a finalized block}.  
Relative to any finalized ledger history, every observer can
objectively determine the resulting account and how much stake is confiscated.
Note that the confiscated amount may depend on the severity of the evidence.
Some minor mistake may only lead to a minor fractional slashing,
while some major violation, e.g., equivocation, may lead to a slashing of whole balance $s_i$.

In most current blockchains, evidence of equivocation is simply a pair of
signatures showing that a validator signed two conflicting blocks, so its
generation requires no interaction among honest validators.  We nevertheless
allow $\mathsf{GenEv}$ to be interactive to keep the definition general.

\subsection{Monopoly Threshold}


One threshold induced by the blockchain protocol is critical to our analysis.
Intuitively, our game-theoretic analysis relies on the coalition's ability to
control enough stake in the underlying consensus
such that the coalition can unilaterally finalize two conflicting blocks,
i.e., equivocating,
and can selectively include or exclude transactions from the finalized blockchain,
thereby excluding any slashing transaction from being finalized.

The monopoly threshold is intended to measure the consensus power conferred
by coalition weight.
Without restricting its strategy, even a coalition of
tiny weight might compromise the protocol by, for example, forging
signatures.  We therefore parameterize the threshold by an admissible strategy
space $\mathcal{S}$, e.g., polynomial-time strategies.
In the next section, we will show that even if the adversary only has access to a further restricted strategy space, 
equivocating two conflicting blocks still forms a Nash equilibrium.

Formally, given a coalition $X\subseteq\players$, its total weight is $W(X)=\sum_{i\in X}w_i$.
Let $\mathcal W=\{W(X):X\subseteq\players\}$ be the finite set of possible coalition weights.  

\begin{definition}[Monopoly threshold]
\label{def:monopoly-thresholds}
Fix a set $\mathcal{S}$ of strategies.
The \emph{monopoly threshold} $q_M$ with respect to $\mathcal{S}$ is the smallest $q\in\mathcal W$ such that,
for every height $h$, every protocol context $\mathsf{ctx}_h$ at height $h$, and every coalition $M\subseteq\players$ with $W(M) \geq q$,
there exists a strategy $S \in \mathcal{S}$ for $M$ such that the following
properties hold with probability one.
\begin{itemize}
  \item \textbf{Equivocation:}
  Given every pair $B,B'$ satisfying $B\perp_h B'$,
  $S$ produces $\varphi,\varphi'$ such that
  \[
  \mathsf{VfyFinal}(\mathsf{ctx}_h,B,\varphi)
  =\mathsf{VfyFinal}(\mathsf{ctx}_h,B',\varphi')=\acc.
  \]
  \item \textbf{Censorship:}
  Given any two disjoint finite sets
$\mathcal T_{\mathsf{in}}$ and $\mathcal T_{\mathsf{out}}$ of transactions made
available to all honest validators at any time during the modeled episode,
every transaction in $\mathcal T_{\mathsf{in}}$ is included
in some block finalized from $\mathsf{ctx}_h$ onward on every branch the coalition finalizes, and no block finalized
from $\mathsf{ctx}_h$ onward contains a transaction in $\mathcal T_{\mathsf{out}}$
during the modeled episode.
\end{itemize}
\end{definition}

If the coalition controls the entire committee, i.e., $M = \players$, then it satisfies both the equivocation and censorship properties
for any natural blockchain protocol.
In such case, the monopoly threshold exists.


\begin{remark}[Computational monopoly threshold]
\label{rem:computational-monopoly-threshold}
\Cref{def:monopoly-thresholds} requires a single admissible strategy
to satisfy equivocation and censorship jointly with probability one.
When the admissible strategies are polynomial-time,
we may instead require these properties to hold jointly with probability
at least $1-\eta_M(\lambda)$,
where $\eta_M$ is negligible in the security parameter $\lambda$.
All other requirements and quantifiers remain unchanged.
We call the resulting threshold the \emph{computational monopoly threshold}.

This relaxation accommodates protocols whose finality guarantees
have negligible error.
Its effect on the equilibrium guarantee of \cref{thm:risk-free}
is discussed in the equilibrium analysis.
\end{remark}

\section{The Basic Coordination Game}
\label{sec:basic-coordination-game}

\subsection{Rational Validators and Equivocation Slashing}

In this section, we consider a concrete slashing mechanism commonly used in
practice.  Given public protocol messages, $\mathsf{GenEv}$ outputs
$(\mathsf{ctx}_h,i,\mathsf{ev}_i)$, where $\mathsf{ev}_i$ consists of
validator $i$'s signatures on two conflicting blocks $B\perp_h B'$ in the
same context $\mathsf{ctx}_h$.
$\mathsf{VfyEv}(\mathsf{ctx}_h,i,\mathsf{ev}_i)$ accepts if the blocks are
valid in that context, they conflict, and both signatures verify under
$\mathsf{pk}_i$.  If the evidence is accepted and fresh,
$\mathsf{ApplySlash}$ sets $i$'s current stake balance to zero and records the
violation in the account's auxiliary field.  

Let $\Pi_{\mathsf{bc}}$ be a blockchain protocol with the 
slashing mechanism described above.  Fix committee
$\players=\{1,\ldots,n\}$, registry $\mathcal V_{\mathsf{bc}}$, initial
validator accounts $\mathsf{acct}_i=(\mathsf{pk}_i,w_i,s_i,\mathsf{aux}_i)$,
and coalition-weight function $W(\cdot)$ for the initial committee episode.
The committee and consensus weights remain fixed within that episode in each
blockchain history.  Their subsequent evolution follows the underlying
protocol.
We assume that at the beginning of the execution, no valid
evidence of an earlier violation is available against any validator.

At the beginning of the game, each validator learns the public parameters of the protocol,
and its own type $t_i$ which is either $\honest$ or $\rational$.
Each validator does not know the types of the other validators.
If a validator has type $\honest$, it always follows the protocol $\Pi_{\mathsf{bc}}$ and the prescribed slashing-enforcement routines.
In particular, it never registers with the coordination mechanism or 
produces conflicting authenticated consensus messages.
If a validator has type $\rational$, it will act to maximize its expected utility, potentially deviating from the protocol if it is beneficial.
Rational validators need not trust one another, and each makes decisions based on its own type and the observed state of the game.
Let $t=(t_1,\ldots,t_n)$ denote the type profile of all validators,
and define \[
  H(t)=\{i:t_i=\honest\},
  \qquad
  R(t)=\{i:t_i=\rational\}.
\]
Nature may choose any profile in $\{\honest,\rational\}^n$. 
There is no probability distribution or common prior over type profiles.

\paragraph{Admissible strategies.}
A rational validator $i\in\players$ may combine the following actions:
\begin{itemize}
  \item Deploy arbitrary smart contracts or submit arbitrary transactions to the blockchain.
  \item As a block proposer, omit a given transaction from its proposed block, or,
  in consensus, decline to vote for a block containing that transaction.
  Effectively, it attempts to censor specific transactions.
  \item Vote for any block in consensus under its public key $\mathsf{pk}_i$, potentially
  deviating from the prescribed protocol.
  Crucially, it may vote for two conflicting blocks, potentially forking the consensus.
  \item Any action prescribed by the honest protocol $\Pi_{\mathsf{bc}}$,
  including honest voting and block proposal in consensus, and the slashing of misbehaving validators.
\end{itemize}
Let $\mathcal S_\text{coord}$ be the admissible strategy space consisting of
probabilistic interactive algorithms built from these actions.  They act on
finite local observations.  At each interaction, computation is polynomial
in the security parameter and the length of the local history.  Validators
control only their own signing keys; strategies that break cryptographic
primitives, such as forging another validator's signature, are excluded.
Let $q_\text{coord}$ denote the monopoly threshold with respect to
$\mathcal S_\text{coord}$, as in
\cref{def:monopoly-thresholds}.

\subsection{Coalition Coordination}
\label{sec:coordination}

\paragraph{The coordination smart contract.}
The coordination mechanism is a smart contract with read access to the fixed
validator registry $\mathcal V_{\mathsf{bc}}$.  
The public registration deadline $\Delta$ is a block height within that episode.  
Let $\mathsf{reg}$ be an instance-specific message that expresses intent to enroll during this
registration period.
A deterministic verifier $\mathsf{VfyReg}(\mathsf{pk}_i,\mathsf{reg},\sigma_i)
  \in\{\acc,\rej\}$
checks whether $\sigma_i$ is a valid signature on $\mathsf{reg}$ under
$\mathsf{pk}_i$.

\begin{mdframed}
  \begin{center}
    \textbf{Coordination Smart Contract}
  \end{center}

\noindent Parameters: $\mathcal V_{\mathsf{bc}}$, $\Delta$, $\mathsf{reg}$, $q_\text{coord}$, $C\gets \emptyset$.

\begin{description}[font=\normalfont]
  \item[\textsc{Register}$(\mathsf{pk}_i,\sigma_i)$.]
  Accept if the call executes at a block height below $\Delta$, the key
  belongs to the registry, $i\notin C$, and
  \[
    \mathsf{VfyReg}(\mathsf{pk}_i,\mathsf{reg},\sigma_i)=\acc.
  \]
  On acceptance, set $C\gets C\cup\{i\}$.

  \item[\textsc{Close}.]
  Accept if the call executes at a block height at least $\Delta$, publish
  $\mathsf{abort}$ if $W(C)<q_\text{coord}$,
  and $\mathsf{activate}$ if $W(C)\ge q_\text{coord}$.
\end{description}
\end{mdframed}

The contract reads each weight from the registry rather than from the caller.
The function \textsc{Close} is permissionless and may be called any number of
times at block heights at least $\Delta$.  In each finalized blockchain
history, its first accepted \textsc{Close} call determines the canonical
outcome and enrollment set $C$.  A validator using the prescribed strategy
acts on the first such
finalized outcome it observes, in the protocol's local processing order, and
ignores later outcomes.
We assume that all validators observe the finalized
\textsc{Close} outcome immediately upon finalization.

In this section, the only evidence that leads to slashing is a validator's 
signatures on the two conflicting blocks.
Therefore, if a validator only signs registration messages and submits calls to the coordination contract, 
and otherwise follows $\Pi_{\mathsf{bc}}$, 
then no admissible algorithm given any finite portion of the execution
transcript can produce evidence accepted by $\mathsf{VfyEv}$ against the
validator.

\paragraph{Prescribed coordination strategy.}
The contract records enrollment and signals only whether coordination aborts
or activates.
Given a context $\mathsf{ctx}_h$, and a set of validators $C$,
let $\beta$ be the fixed deterministic function used to select conflicting blocks:
  \[
    \beta(\mathsf{ctx}_h,C)\coloneqq(B_C,B'_C),
    \qquad B_C\perp_h B'_C.
  \]
The complete prescribed coordination strategy for a rational validator $i$
is as follows.
\begin{enumerate}[leftmargin=2.4em]
  \item \textbf{Register.}
  At the start of the registration period, sign $\mathsf{reg}$ under
  $\mathsf{pk}_i$ and submit a \textsc{Register} call.
  Otherwise follow $\Pi_{\mathsf{bc}}$.
  \item \textbf{Close and observe.}
  At a block height at least $\Delta$, call \textsc{Close} if no finalized
  outcome is yet observed, then continue to follow $\Pi_{\mathsf{bc}}$
  while awaiting that outcome.
  If its status is $\mathsf{abort}$ or is $\mathsf{activate}$ but $i\notin C$,
  take no attack action and continue to follow                          
  $\Pi_{\mathsf{bc}}$, with no withdrawal requirement.
  The remaining steps apply only after observing a finalized $\mathsf{activate}$
  with $i\in C$.
  \item \textbf{Select the attack.}
  Let $h$ be the earliest block height following the finalized
  \textsc{Close} call at which the coordinated attack can begin.
  Individually compute the selected pair of conflicting blocks $(B_C,B'_C) = \beta(\mathsf{ctx}_h,C)$.

  For each activated $C$ and context $\mathsf{ctx}_h$,
  fix a strategy $S_C=(S_{C,i})_{i\in C}$ witnessing the monopoly property
  in \cref{def:monopoly-thresholds} for the selected pair $(B_C,B'_C)$.
  \item \textbf{Execute the continuation.}
  Following $S_{C,i}$, publish conflicting authenticated consensus messages
  supporting $B_C$ and $B'_C$.
  Request withdrawal in every finalized blockchain history in which
  $i$'s stake remains bonded, including both selected branches.
  Furthermore, exclude all transactions that would confiscate stake from
  any $j\in C$, and include all other transactions,
  including the coalition's withdrawal requests in particular.
  \item \textbf{Continue until withdrawal.}
  Continue consensus participation in each such history until $i$'s
  withdrawal takes effect, helping advance that history to the episode
  boundary while maintaining transaction inclusion and censorship.
\end{enumerate}

Registration is nonbinding.  Signing and submitting $\mathsf{reg}$ alone
creates no valid slashing evidence.  The contract locks no collateral, so an
abort itself causes no modeled loss.

\subsection{Coordination Game}
\label{sec:coordination-game}

\paragraph{Players and strategies.}
Set wall-clock time zero at
the start of the current committee episode and fix a public game deadline $T>0$.  
We represent play through time $T$ as a normal-form game of complete
strategies.  A strategy specifies a validator's behavior after every possible
sequence of its observations, including whether and when to register and how
to participate in the subsequent blockchain execution.  It may depend on the
validator's own type, received messages, and observed blockchain state, but
not on the unobserved types of other validators.  Thus, a strategy can respond
to later information; choosing it does not bind the validator to the
prescribed coordination protocol.

Let $\mathcal S_i^{\rational}$ contain all admissible complete strategies for
validator $i$, including randomized strategies, and let
$\mathcal S_i^{\honest}$ contain only its prescribed honest strategy. 
Write $s^{\mathsf{atk}}$ for the profile in which rational validators use the coordination strategy
described in \cref{sec:coordination} and honest validators follow $\Pi_{\mathsf{bc}}$.

\paragraph{How the game proceeds.}
At wall-clock time zero, the execution begins with a public blockchain
context containing the registry and the deployed coordination contract.
The game ends at time $T$, even if block production stalls or some histories
have not reached the episode boundary.  Message delivery, transaction
ordering, and consensus execution follow the underlying protocol and its
network model.

We assume that $T$ is sufficiently large to allow the protocol to reach the initial episode boundary under normal network conditions
if the blockchain's height progresses according to the protocol.
This timing assumption applies to both selected attack branches.  

At block heights below $\Delta$, validators may submit registration calls.
Finalization makes the resulting enrollment state observable in that blockchain history.  
At block heights at least $\Delta$, any validator may submit
\textsc{Close}.  Its finalized outcome reports the enrollment set $C$ and
either $\mathsf{abort}$, when
$W(C)<q_\text{coord}$, or $\mathsf{activate}$, when
$W(C)\ge q_\text{coord}$.  
After activation, enrollees may follow
$S_C=(S_{C,i})_{i\in C}$ or deviate.  Rational validators, whether enrolled
or not, retain their full admissible strategy space throughout the execution,
e.g., sending arbitrary transactions to the blockchain and participating in consensus according to their own strategies.

We assume that when all validators use $s^{\mathsf{atk}}$, every rational
validator's registration executes at a block height below $\Delta$, and all
these registrations and a canonical \textsc{Close} call finalize in a
common blockchain history.  Validators observe this outcome early enough
for the prescribed continuation, including the attack and withdrawal
requests before the episode boundary, and completion of withdrawal before
$T$ when the selected branches progress as stipulated above.

Conflicting consensus messages allow observers to generate slashing evidence
and submit transactions invoking $\mathsf{ApplySlash}$.  
Honest validators follow the prescribed slashing mechanism and submit transactions invoking $\mathsf{ApplySlash}$.
Rational validators may report
evidence or attempt to censor its execution.  
A penalty takes effect only when its transaction executes in a finalized block.

\paragraph{Outcomes.}
Given a strategy profile $s$ and
type profile $t$ induce a distribution over executions $\omega$ during the
wall-clock interval $[0,T]$.  An outcome records messages, finality certificates,
and executed state transitions across all blockchain histories produced by
that time. Reaching an episode
boundary on one or more forks does not end the game before $T$.

We define three events of interest for each validator $i$.
Let $\Attack_i(\omega)$ denote the event that some finalized blockchain history
in $\omega$ has a canonical \textsc{Close} outcome $\mathsf{activate}$
with enrollment set $C$ containing $i$ and associated attack context
$\mathsf{ctx}_h$ in the initial committee episode, and $\omega$ contains
certificates $\varphi,\varphi'$ produced at or before $T$ satisfying
\[
  \mathsf{VfyFinal}(\mathsf{ctx}_h,B_C,\varphi)
  =\mathsf{VfyFinal}(\mathsf{ctx}_h,B'_C,\varphi')=\acc,
  \qquad (B_C,B'_C)=\beta(\mathsf{ctx}_h,C).
\]

Let $\mathsf{Slashed}_i(\omega)$ denote the event that, during $[0,T]$, an
$\mathsf{ApplySlash}$ transition against $i$ executes in at least one
finalized blockchain history and confiscates a positive amount of stake.
Evidence alone, or an unfinalized slashing transaction does not qualify.
Withdrawal on another fork does not erase an executed penalty.

For each successful coordinated attack enrolling $i$, withdrawal must be
completed on both selected branches by $T$.  On each branch, withdrawal is
complete when a finalized block releases $i$'s
full initial stake $s_i$.  The event $\mathsf{ExitFailed}_i(\omega)$ holds
if the withdrawal is not finalized on at least one selected branch by $T$.

All three events are evaluated at $T$.  Events occurring after $T$ do not
affect the game's outcome or utility.

\paragraph{Utility.}
We assume that $\beta(\mathsf{ctx}_h,C)$ selects conflicting blocks whose
successful finalization by $T$ gives each enrollee in $C$ a total gain of
$\epsilon>0$, e.g., through double-spending,
even if that enrollee deviates from $S_C$ or free rides.
This gain is counted once even if several activated outcomes succeed.
If $i$ is slashed or fails to complete the required withdrawals, the loss is
$i$'s initial stake $s_i$, charged once even when both conditions hold or
several finalized histories contain slashing.  Returning the initial stake
adds no separate gain.

Validator $i$'s utility is defined as
\begin{equation}
  \label{eq:utility}
  u_i(\omega)
  =\epsilon\ind\{\Attack_i(\omega)\}
  -s_i\ind\{\mathsf{Slashed}_i(\omega)
    \lor\mathsf{ExitFailed}_i(\omega)\}.
\end{equation}
Given a strategy profile $s$ and type profile $t$, the payoff in $G(t)$ is
\[
  U_i(s;t)=\mathbb E_{\omega\mid s,t}[u_i(\omega)],
\]
where the expectation is over the internal randomness of the validators and
the protocol.

This is a pessimistic terminal valuation of stake.  A penalty executed in
any finalized blockchain history by $T$ incurs the full loss.  A successful
attacker also incurs the full loss if either selected branch lacks a
completed withdrawal by $T$, even if no penalty executes.  In particular,
stalling before the episode boundary does not avoid this loss after a
successful attack.

For simplicity, submitted transactions incur no fees, and capital use has no
direct or opportunity cost beyond the explicit terminal loss convention
above.  A validator enrolled in no successful selected attack has neither
an attack gain nor a withdrawal-related loss, but may still incur a slashing
loss for independent misconduct.


\subsection{Equilibrium Analysis}

\subsubsection{Risk-free attack equilibrium}

For $\widehat{s}_i\in\mathcal S_i^{\rational}$, write
$(\widehat{s}_i,s_{-i})$ for the profile obtained by replacing only $i$'s
rational-type strategy.

\begin{definition}[Ex post Nash equilibrium]\label{def:epne}
A type-contingent profile $s^\star$ is an \emph{ex post Nash equilibrium} if,
for every $t\in\{\honest,\rational\}^n$, every $i\in R(t)$, and every
admissible complete strategy $\widehat{s}_i\in\mathcal S_i^{\rational}$,
\[
  U_i(s^\star;t)
  \ge U_i\bigl((\widehat{s}_i,s_{-i}^\star);t\bigr).
\]
Alternative strategies may be randomized.  Expectations are over the
randomness of the validators and the protocol, with $t$ fixed.  Thus the same
profile is a Nash equilibrium in every $G(t)$, without a prior over types.
\end{definition}

In our model, we define the termination of the game with respect to the 
wall-clock time $T$,
and the outcome of the game is defined with respect to the finalized blockchain history at time $T$.
Therefore, to formally prove that the prescribed strategy profile constitutes a Nash equilibrium,
we assume that the blockchain progresses in a timely manner so that all prescribed transactions can be included and finalized,
which is formulated as follows.

\begin{assumption}[Timely coordinated execution]
\label{ass:timely-execution}
For every type profile $t$, the following hold with probability one.
Under prescribed play, every transaction prescribed by the
coordination strategy executes in a finalized block,
in the prescribed order and within its protocol deadlines,
on every blockchain branch.
All blockchain progress needed to submit and execute these
transactions also occurs early enough for the prescribed
coordination process to complete by $T$.

If coordination activates with coalition $C\subseteq R(t)$
satisfying $W(C)\ge q_{\text{coord}}$,
and every member of $C$ follows its prescribed coordination strategy,
the same continuation guarantees hold for $C$ regardless of
the admissible strategies of validators outside $C$.
\end{assumption}

We remark that this assumption is mild. 
Before activation, the prescribed transactions are ordinary registration and \textsc{Close} calls, which honest validators include under the liveness of $\Pi_{\mathsf{bc}}$.
If the coalition already controls
at least the monopoly threshold, i.e., $W(C)\ge q_{\text{coord}}$,
it is justified for the coalition to complete all prescribed coordination transactions in a timely manner.

\begin{theorem}[Risk-free attack equilibrium]\label{thm:risk-free}
Under \cref{ass:timely-execution}, the prescribed profile
$s^{\mathsf{atk}}$ is an ex post Nash equilibrium.  
Moreover, for every
$t\in\{\honest,\rational\}^n$ and every rational validator $i\in R(t)$,
\[
  U_i(s^{\mathsf{atk}};t)=
  \begin{cases}
    0, & W(R(t))<q_{\mathrm{coord}},\\
    \epsilon, & W(R(t))\ge q_{\mathrm{coord}}.
  \end{cases}
\]
Along this equilibrium, no rational validator is slashed in any finalized
blockchain history during $[0,T]$.  If $R(t)\ne\varnothing$ and
$W(R(t))\ge q_{\mathrm{coord}}$, the selected attack succeeds and every
rational validator completes all required full-stake withdrawals by $T$.
\end{theorem}

\begin{proof}
Fix $t$ and $i\in R(t)$.
If all rational validators follow $s^{\mathsf{atk}}$ and register,
we have $C=R(t)$, since honest validators never register.

Suppose $W(R(t))<q_{\mathrm{coord}}$.  The contract aborts.  No validator
equivocates, registration creates no slashing evidence, and no withdrawal is
required, so $U_i(s^{\mathsf{atk}};t)=0$.  
Under any unilateral admissible deviation, 
the coalition's weight is below $q_{\mathrm{coord}}$, so activation
and a coordinated attack gain are impossible.  The deviator's utility is at
most zero by~\cref{eq:utility}, and its payoff cannot improve.

Suppose instead $W(R(t))\ge q_{\mathrm{coord}}$.
The prescribed $S_C$ witnesses the monopoly property for
$C$ and the selected pair $(B_C,B'_C)$.
By \cref{def:monopoly-thresholds}, with probability one it produces
conflicting finality and excludes every transaction that would confiscate
coalition stake.
This censorship applies in every finalized blockchain history through
finalization of that history's initial episode boundary.
By \cref{ass:timely-execution}, both selected branches reach the
initial episode boundary by $T$.
Moreover, in each blockchain history, a withdrawal request is finalized before the episode closes, and it releases the validator's remaining stake when that boundary is finalized. 
The validator retains its consensus weight until
then, and its departure takes effect at that boundary,
so the selected attack succeeds by $T$.

Before the attack begins, prescribed play produces no slashing evidence.
Thereafter, consider any finalized blockchain history through $T$.
If it remains in the initial episode, censorship prevents slashing
throughout the remaining game.
If its initial episode boundary is finalized by $T$,
\cref{ass:timely-execution} ensures that every coalition member's
withdrawal request executes before that boundary.
Censorship preserves each member's stake until the boundary,
at which the withdrawal releases the validator's full stake.
The prescribed strategy leaves that stake unbonded,
so no later slashing transition can confiscate it in this history.
Thus no rational validator is slashed in any finalized blockchain history
during $[0,T]$.
In particular, every rational validator completes its full-stake withdrawal
on both selected branches by $T$.
Therefore $U_i(s^{\mathsf{atk}};t)=\epsilon$.
\cref{eq:utility} bounds every deviating outcome's utility by
$\epsilon$, so no unilateral strategy, including a randomized one, can
increase the expected payoff.  Both cases establish the equilibrium
inequalities for every type profile and rational validator.
\end{proof}

\begin{remark}[Overwhelming-probability guarantees]
\label{rem:computational-risk-free}
Suppose that the monopoly properties hold jointly with failure probability
at most $\eta_M(\lambda)$,
and the timely-execution guarantees of \cref{ass:timely-execution}
hold with failure probability at most $\eta_T(\lambda)$,
where both functions are negligible.
Use the corresponding computational monopoly threshold
as $q_{\mathrm{coord}}$, and set
$\eta(\lambda)=\eta_M(\lambda)+\eta_T(\lambda)$.

Below the threshold, prescribed play still gives payoff zero,
and no unilateral deviation improves the payoff.
At or above the threshold,
a union bound shows that the selected attack succeeds,
no rational validator is slashed,
and all required withdrawals complete by $T$,
jointly with probability at least $1-\eta(\lambda)$.
On this event, validator $i$ receives utility $\epsilon$;
otherwise, its utility is at least $-s_i$.
Thus
\[
  \epsilon-(\epsilon+s_i)\eta(\lambda)
  \le U_i(s^{\mathsf{atk}};t)\le\epsilon.
\]
Since every deviating outcome has utility at most $\epsilon$,
for every type profile $t$, rational validator $i\in R(t)$,
and admissible unilateral deviation $\widehat s_i$,
\[
  U_i((\widehat s_i,s_{-i}^{\mathsf{atk}});t)
  \le U_i(s^{\mathsf{atk}};t)
    +(\epsilon+s_i)\eta(\lambda).
\]
Consequently, the prescribed profile is an ex post Nash equilibrium
up to a negligible additive error,
provided $\epsilon$ and the stakes are polynomially bounded in $\lambda$.
Above the threshold, participation remains strictly preferable
to honest behavior without registration whenever
$\epsilon>(\epsilon+s_i)\eta(\lambda)$.
For every fixed $\epsilon>0$ and polynomially bounded stakes,
this condition holds for sufficiently large $\lambda$.
\end{remark}

\subsubsection{Strict incentive to join}
We show that once other rational validators coordinate,
an individual validator is incentivized to join the coordination.

\begin{corollary}\label{cor:join-incentive}
Under the assumptions of \cref{thm:risk-free}, fix a type profile
$t\in\{\honest,\rational\}^n$ with $W(R(t))\ge q_{\mathrm{coord}}$ and
$i\in R(t)$.  Let $h_i\in\mathcal S_i^{\rational}$ be the strategy that
follows $\Pi_{\mathsf{bc}}$ and never registers.  
Then \[
  U_i(s^{\mathsf{atk}};t)
  =\epsilon
  >0
  =U_i\bigl((h_i,s_{-i}^{\mathsf{atk}});t\bigr).
  \]
Thus honest behavior is not a best response to $s_{-i}^{\mathsf{atk}}$.
\end{corollary}

\begin{proof}
Under $h_i$, validator $i$ never signs a registration message,
so $i$ is outside every enrollment set in every finalized
blockchain history, so both $\Attack_i$ and $\mathsf{ExitFailed}_i$ are
false.  Honest protocol behavior yields no accepted slashing evidence against
$i$, so $\mathsf{Slashed}_i$ is also false throughout $[0,T]$.
\cref{eq:utility} therefore gives
$U_i((h_i,s_{-i}^{\mathsf{atk}});t)=0$.  Switching only $i$ to its prescribed
coordination strategy yields $s^{\mathsf{atk}}$, which gives $i$ payoff
$\epsilon>0$ by \cref{thm:risk-free}.
\end{proof}

\section{Extensions of the Coordination Game}\label{sec:variants}

We give two extensions of the basic coordination game,
which demonstrate the feasibility of risk-free coordination among validators
even if the slashing mechanism becomes more general.
In \cref{sec:evidence-rewards}, we reward validators for reporting equivocation.
We show that if validators deposit sufficient collateral, 
the prescribed coordination strategy remains an ex post Nash equilibrium.
In \cref{sec:anon-collateral}, we consider slashing mechanisms that may punish behavior beyond conflicting signatures.
We propose an anonymous enrollment contract that allows validators to register without revealing their identities, 
protecting validators from slashing before actual equivocation.

\subsection{Coordination with Reporting Rewards}\label{sec:evidence-rewards}

Extend the game in \cref{sec:coordination-game} by allowing rewards for
reporting equivocation. For each validator $j$, fix a public reward budget
$0\le b_j\le s_j$. Upon slashing, up to $b_j$ of $j$'s confiscated initial
stake is paid to reporters, and the remainder is burned.
The total reward funded by $j$'s initial stake across all reporters and
blockchain histories is at most $b_j$.
A reward is paid
only when the corresponding slashing transition executes in a finalized
block. 
Withdrawn initial stake cannot fund a later reward in the same history.

\paragraph{Collateralized coordination contract.}
In the collateralized variant, each validator must deposit a public collateral amount $D_i>0$ 
separate from its slashable stake $s_i$ into the smart contract to register.
The contract retains the registry $\mathcal V_{\mathsf{bc}}$, deadline $\Delta$ counted in block height, and threshold
$q_\text{coord}$ of \cref{sec:coordination}. It can also verify records
of executed slashing and full-stake withdrawal in its own finalized
blockchain history. The selected context and conflicting pair
$(B_C,B'_C)=\beta(\mathsf{ctx}_h,C)$ are the same as \cref{sec:coordination}.

\begin{mdframed}
  \begin{center}
    \textbf{Collateralized Coordination Smart Contract}
  \end{center}

\noindent Parameters: 
$\mathcal V_{\mathsf{bc}}$, $\Delta$, $\mathsf{reg}$, $q_\text{coord}$, and collateral amounts $(D_i)_{i\in\players}$. 
Initially, $C\gets \emptyset$,

\begin{description}[font=\normalfont]
  \item[\textsc{Register}$(\mathsf{pk}_i,\sigma_i)$.]
  Accept if the call executes at a block height below $\Delta$,
  the key belongs to $\mathcal V_{\mathsf{bc}}$, $i\notin C$,
  $\mathsf{VfyReg}(\mathsf{pk}_i,\mathsf{reg},\sigma_i)=\acc$, and the call
  deposits exactly $D_i$.
  On acceptance, lock $D_i$ and add $i$ to $C$.

  \item[\textsc{Close}.]
  At a block height at least $\Delta$, the first accepted call closes
  registration and freezes $C$. If $W(C)<q_\text{coord}$, publish
  $\mathsf{abort}$ and return every deposit in this same execution.
  Otherwise, publish $\mathsf{activate}$ and retain all deposits. Later
  calls leave the outcome and balances unchanged.

  \item[\textsc{Settle}$(i,\mathsf{ev}_i)$.]
  After activation, first check whether any $j\in C$ has incurred a
  positive slashing penalty in this blockchain history. If so, burn every
  outstanding deposit, regardless of $i$ and $\mathsf{ev}_i$. Otherwise,
  return $D_i$ to validator $i$ if its deposit is outstanding,
  $\mathsf{ev}_i$ contains its valid signatures on both selected blocks
  $B_C,B'_C$ in $\mathsf{ctx}_h$, and every $j\in C$ has completed
  withdrawal of its full initial stake $s_j$ in this history. In all other
  cases, leave the deposits unchanged.
\end{description}
\end{mdframed}

Both \textsc{Close} and \textsc{Settle} are permissionless. Refunds always
return funds to the enrolled validator, not to the caller. 
Once a coalition member is slashed in a history, the refund condition is
permanently false there, even if nobody calls \textsc{Settle}. Conversely,
after every member has withdrawn its full initial stake without slashing,
that stake can no longer generate a reporting reward in that history.
Consequently, no additional reporting deadline is needed before release.

\paragraph{Prescribed strategy.}
Let $s^{\mathsf{col}}$ denote the following collateralized modification
of the baseline profile $s^{\mathsf{atk}}$.
Each rational validator $i$ follows these steps.
\begin{enumerate}[leftmargin=2.4em]
  \item \textbf{Register.}
  Deposit $D_i$ with the registration call and otherwise follow the
  baseline registration strategy.
  \item \textbf{Close and observe.}
  Follow the baseline \textsc{Close} strategy.
  On abort, the collateral is returned;
  continue to follow $\Pi_{\mathsf{bc}}$.
  The remaining steps apply only upon activation with $i\in C$.
  \item \textbf{Attack and withdraw.}
  Follow $S_{C,i}$ as defined in \cref{sec:coordination},
  including publishing conflicting signatures, requesting withdrawal, 
  censoring slashing transactions, and advancing the blockchain.
  Importantly, do not report coalition members.
  \item \textbf{Settle collateral.}
  Submit \textsc{Settle} calls with $i$'s conflicting signatures once
  the refund conditions hold on each branch.
\end{enumerate}
All other contingencies are handled as in the baseline complete strategy.

\paragraph{Collateral valuation and utility.}
Given a blockchain execution history $\omega$,
let $r_i(\omega)\ge0$ be the total reporting reward received by validator
$i$, including through accounts it controls. 
Define $\mathsf{BondLost}_i(\omega)$ to hold if, at the end of the game
at time $T$, validator $i$'s collateral has not been returned in full
on at least one finalized blockchain branch containing its accepted deposit.
The event is false if no deposit by $i$ is finalized.

Using the baseline utility $u_i$ from~\cref{eq:utility}, define
\begin{equation}\label{eq:reward-collateral-utility}
  u_i^{D,b}(\omega)
  =u_i(\omega)+r_i(\omega)
  -D_i\ind\{\mathsf{BondLost}_i(\omega)\}.
\end{equation}
Write $U_i^{D,b}(s;t)=\mathbb E_{\omega\mid s,t}[u_i^{D,b}(\omega)]$.
Returning collateral adds no gain. 

The baseline registration and timing assumptions,
including \cref{ass:timely-execution}, apply to $s^{\mathsf{col}}$.


\begin{theorem}[Collateralized attack equilibrium with evidence rewards]
\label{thm:reward-equilibrium}
Under \cref{ass:timely-execution},
suppose that, for every validator $i$,
\begin{equation}\label{eq:reward-collateral-bound}
  D_i>0,
  \qquad D_i\ge\sum_{j\in\players\setminus\{i\}}b_j,
\end{equation}
and suppose that the blockchain does not fork at the execution of \textsc{Close}, 
so that the outcome of \textsc{Close} is the common prefix of all possible finalized blockchain branches.

Then, $s^{\mathsf{col}}$ is an ex post Nash equilibrium of the
game with utility as per \labelcref{eq:reward-collateral-utility}. 
Moreover, for every type
profile $t$ and every $i\in R(t)$,
\[
  U_i^{D,b}(s^{\mathsf{col}};t)
  =\begin{cases}
    0, & W(R(t))<q_\text{coord},\\
    \epsilon, & W(R(t))\ge q_\text{coord}.
  \end{cases}
\]
Along prescribed play, every rational validator recovers its collateral
and none is slashed. When $R(t)\ne\emptyset$ and
$W(R(t))\ge q_\text{coord}$, the selected attack succeeds and every
rational validator completes all required full-stake withdrawals by $T$.
\end{theorem}

\begin{proof}
Fix $t$ and $i\in R(t)$.
Under $s^{\mathsf{col}}$, prescribed registration gives $C=R(t)$.
Suppose $W(C)<q_{\mathrm{coord}}$.
By \cref{ass:timely-execution}, \textsc{Close} finalizes with an abort
outcome by $T$ on every branch containing an accepted deposit,
returning that deposit.
No validator equivocates, so no slashing or reporting reward arises.
Thus $U_i^{D,b}(s^{\mathsf{col}};t)=0$.

Under any unilateral admissible deviation,
the coalition's weight is always below $q_{\mathrm{coord}}$, so $\Attack_i$ is false.
Since no evidence of earlier violations is available and signatures cannot
be forged, any reward received by $i$ must arise from slashing $i$ itself, 
giving $r_i\le b_i\ind\{\mathsf{Slashed}_i\}$.
Since $b_i\le s_i$ and collateral losses are nonnegative,
the deviator's utility is at most zero.

Suppose instead $W(C)\ge q_{\mathrm{coord}}$.
The argument in the proof of \cref{thm:risk-free} applies to the prescribed
attack continuation: with probability one, the selected attack succeeds,
no rational validator is slashed in any finalized blockchain history during $[0,T]$,
and every rational validator completes its required full-stake withdrawals by $T$.
By \cref{ass:timely-execution}, the prescribed withdrawals and subsequent
eligible \textsc{Settle} calls complete on every blockchain branch by $T$,
returning all collateral.
No reporting reward is paid, so $U_i^{D,b}(s^{\mathsf{col}};t)=\epsilon$.

Now consider any unilateral admissible deviation by $i$.
Suppose $i$ receives a reward for slashing some $j\ne i$.
Unforgeability and the absence of earlier evidence imply that $j$
produced conflicting signatures in a prescribed attack.
Let $C$ now denote that attack's activated enrollment set,
so $j\in C$ and $W(C)\ge q_{\mathrm{coord}}$.

By the common-prefix premise, nondeviating attackers act on the same
finalized \textsc{Close} outcome, and every history paying rewards
for their attack evidence extends that outcome.
If $i\notin C$, every member of $C$
follows the same prescribed continuation.
Monopoly censorship and timely withdrawal then exclude slashing $j$
by the argument in \cref{thm:risk-free}, contradicting the reward.
If $i\in C$, the reward-paying history contains its deposit.
A prior refund would require $j$'s initial stake to have been withdrawn,
precluding the reward.
The penalty therefore permanently prevents refund of the outstanding
deposit, so $\mathsf{BondLost}_i$ holds.

Self-reporting rewards require $\mathsf{Slashed}_i$ and total at most $b_i$.
Thus the reward budgets and
\cref{eq:reward-collateral-bound} give
\begin{equation}\label{eq:reward-deviation-bound}
  \begin{aligned}
    r_i
    &\leq b_i\ind\{\mathsf{Slashed}_i\}
      +\left(\sum_{j\ne i}b_j\right)
        \ind\{\mathsf{BondLost}_i\}\\
    &\leq s_i\ind\{\mathsf{Slashed}_i\}
      +D_i\ind\{\mathsf{BondLost}_i\}.
  \end{aligned}
\end{equation}
Substituting into~\cref{eq:reward-collateral-utility} yields
\begin{align*}
  u_i^{D,b} \leq
  \epsilon\ind\{\Attack_i\} - s_i \ind\{\mathsf{Slashed}_i\} + r_i - D_i  \ind\{\mathsf{BondLost}_i\} \leq \epsilon.
\end{align*}
Taking expectations gives the equilibrium
inequalities for every type profile and rational validator,
including against randomized unilateral deviations.
\end{proof}

\begin{corollary}\label{cor:reward-join-incentive}
Under the assumptions of \cref{thm:reward-equilibrium}, fix a type profile
$t\in\{\honest,\rational\}^n$ with $W(R(t))\ge q_{\mathrm{coord}}$ and
$i\in R(t)$.
Let $h_i\in\mathcal S_i^{\rational}$ be the strategy that
follows $\Pi_{\mathsf{bc}}$ and never registers.
Then
\[
  U_i^{D,b}(s^{\mathsf{col}};t)
  =\epsilon
  >0
  =U_i^{D,b}\bigl((h_i,s_{-i}^{\mathsf{col}});t\bigr).
\]
Thus honest behavior is not a best response to $s_{-i}^{\mathsf{col}}$,
with $i$'s type remaining rational.
\end{corollary}

\begin{proof}
Under $h_i$, validator $i$ never registers or deposits collateral,
so $\Attack_i$, $\mathsf{ExitFailed}_i$, and $\mathsf{BondLost}_i$ are false.
Honest protocol behavior also makes $\mathsf{Slashed}_i$ false.

Honest behavior may include reporting other validators' equivocations,
but a reporting reward requires an executed slashing penalty.
If coordination does not activate, no validator equivocates,
so no reporting reward arises.
If coordination activates with coalition $C$, then $i\notin C$ and
$W(C)\ge q_{\mathrm{coord}}$.
Every member of $C$ follows the prescribed strategy,
and the common-prefix premise ensures that they act on the same
finalized \textsc{Close} outcome.
By \cref{def:monopoly-thresholds,ass:timely-execution}, coalition censorship
prevents slashing until full withdrawal in each finalized blockchain history.
After withdrawal, that initial stake cannot fund a reporting reward.
Validators outside $C$ follow the honest protocol and do not equivocate.
Thus even if $i$ reports coalition members, no rewarded slashing penalty
executes, and $r_i=0$ with probability one.

Thus \cref{eq:reward-collateral-utility} gives
$U_i^{D,b}((h_i,s_{-i}^{\mathsf{col}});t)=0$.
Switching only $i$ to its prescribed coordination strategy yields
$s^{\mathsf{col}}$, which gives $i$ payoff $\epsilon>0$
by \cref{thm:reward-equilibrium}.
\end{proof}

\subsection{Anonymous Enrollment under General Slashing Rules}
\label{sec:anon-collateral}

In \cref{sec:basic-coordination-game}, slashing evidence consists of a validator's signatures on two conflicting blocks.
A more aggressive mechanism might also recognize other evidence, such as deploying a coordination
contract or registering an intention to equivocate as evidence.
In this section, we consider any slashing mechanism satisfies a natural property called \emph{honest-stake safety}:
it must not permit accepted evidence against an honest validator.

We introduce a coordination protocol with anonymous enrollment. 
Before activation, validator processes follow the honest blockchain protocol, 
while enrollment occurs through user accounts that are unlinkable to validator identities. 
Even if the mechanism considers enrollment punishable, 
its available information is consistent with the target validator being honest and owning no enrollment accounts, 
while other validators operate those accounts.
Thus, attributing enrollment to that validator could potentially punish an honest validator, violating honest-stake safety.
We present the formal definition as follows.

\begin{definition}[Honest-stake safety]
\label{def:honest-stake-safety}
  A slashing mechanism satisfies \emph{honest-stake safety} if,
  for every coalition $X\subseteq\players$, every probabilistic polynomial-time
  strategy of $X$, and every validator $i\notin X$ that follows
  $\Pi_{\mathsf{bc}}$, 
  $X$ does not output a pair
  $(\mathsf{ctx}_h,\mathsf{ev})$ such that $\mathsf{ctx}_h$ is reached in the
  resulting execution and
  $\mathsf{VfyEv}(\mathsf{ctx}_h,i,\mathsf{ev})=\acc$.
\end{definition}

\subsubsection{Coordination Protocol}
\label{sec:anonymous-deposit-protocol}

We retain the committee, registry, deadlines, admissible complete strategies,
and no transaction fee conventions as \cref{sec:basic-coordination-game},
with no reporting rewards or evidence of earlier violations.
For the sake of analyzing anonymity,
we assume all validators have equal weights, $w_i=1/n$, and set $m=\lceil nq_{\mathrm{coord}}\rceil$.

A rational validator can create and use multiple user accounts.
We assume those accounts are generated and managed independently of the validator's identity.
Honest validators never create those additional user accounts.




\paragraph{The smart contract.}
The parameters $\mathcal V_{\mathsf{bc}}$, $\mathsf{reg}$, $\Delta$, and $D$
are defined as in \cref{sec:evidence-rewards}.
A settlement certificate $\mathsf{cert}$ is called \emph{valid} for account $a$ and validator $i$
if $\mathsf{cert}$ contains $i$'s valid signatures on both selected blocks,
and refund authorizations signed by both $a$ and $i$.
The signatures by both $a$ and $i$ establish the link between the user $a$ and the validator $i$.
The contract keeps track of the set of enrolled accounts $A$ and the set of validators $F$ that have used their refund entitlements.

\begin{mdframed}
  \begin{center}
    \textbf{Anonymous Deposit Coordination Contract}
  \end{center}

  \noindent Parameters:
  $\mathcal V_{\mathsf{bc}}$, $\mathsf{reg}$, $\Delta$, $m$, and $D$.
  Initially, $A\gets\emptyset$ and $F\gets\emptyset$.

  \begin{description}[font=\normalfont]
    \item[\textsc{Register}$(a)$.]
    Accept a call from account $a$ if it executes below
    height $\Delta$, $a\notin A$, and transfers exactly $D$
    from $a$ to the contract.
    Lock the deposit and add $a$ to $A$.
    Invalid, duplicate, or late calls leave enrollment unchanged.

    \item[\textsc{Close}.]
    At height at least $\Delta$, the first accepted call freezes $A$.
    If $|A|<m$, publish $\mathsf{abort}$ and return every deposit
    to its original account.
    Otherwise, publish $\mathsf{activate}$ and retain the deposits.
    Earlier or repeated calls leave the outcome and balances unchanged.

    \item[\textsc{Settle}$(a,i,\mathsf{cert})$.]
    After activation, accept if $a\in A$ has an outstanding deposit,
    $i \in \players$,
    $i \notin F$,
    $\mathsf{cert}$ is valid for $a$ and $i$,
    and withdrawal of $i$'s full initial stake $s_i$ has completed in this
    blockchain history.
    Return $D$ only to $a$, mark its deposit returned,
    and add $i$ to $F$.
    Otherwise, leave the state unchanged.
  \end{description}
\end{mdframed}

Both \textsc{Close} and \textsc{Settle} are permissionless.
After activation, deposits that fail to satisfy \textsc{Settle} remain
locked and count as unreturned collateral at $T$.
Since one player may fund many accounts,
even $|A|=n$ does not certify participation by every validator.

Retain the common-prefix premise of \cref{thm:reward-equilibrium}:
all finalized \textsc{Close} outcomes agree on one canonical call,
status, and frozen enrollment set $A$,
including under unilateral deviations.
This permits subsequent forks and imposes no transaction-finalization
guarantee under deviations.

\paragraph{Prescribed strategy.}
For a feasible context $\mathsf{ctx}_h$ and a set $A$ of anonymous
user accounts, let the deterministic equivocation builder select
\[
  \beta^{\mathsf{ano}}(\mathsf{ctx}_h,A)
    =(B_{A},B'_{A}),
  \qquad B_{A}\perp_h B'_{A}.
\]
Finalizing both selected blocks gives each account in $A$ a total gain $\epsilon$;
accounts outside $A$ receive no modeled attack gain.
For each feasible context, selected pair, and set of registered signers
whose total weight is at least $q_{\mathrm{coord}}$,
fix publicly an admissible continuation witnessing the monopoly property
of \cref{def:monopoly-thresholds}.
The prescribed continuation identifies its participants at the first local
observation of valid signature pairs on the selected blocks
from at least $|A|$ distinct registered validators.
It fixes the first $|A|$ observed signers in registry order as its participants.
A validator executes its component only if selected;
otherwise, it waits subject to the prescribed censorship.

Let $s^{\mathsf{ano}}$ modify the basic strategy in \cref{sec:coordination}
as follows for each rational validator $i$.
Before observing finalized activation, the validator process follows
$\Pi_{\mathsf{bc}}$ exactly.
Honest validators retain their honest strategy.
\begin{enumerate}[leftmargin=2.4em]
  \item \textbf{Register.}
  Generate a fresh user account $a_i$ and submit
  $\textsc{Register}(a_i)$ with deposit $D$.

  \item \textbf{Close and observe.}
  At a block height at least $\Delta$, call \textsc{Close} if no finalized
  outcome has been observed, and follow $\Pi_{\mathsf{bc}}$ while awaiting it.
  After abort, or activation with $a_i\notin A$,
  follow the honest protocol without revealing the account-to-validator link.
  Upon observing activation with $a_i\in A$, immediately begin censoring
  every transaction that would slash a validator in $\mathcal V_{\mathsf{bc}}$.
  Maintain this censorship throughout the episode.

  \item \textbf{Prepare and sign.}
  Let $h$ be the earliest block height following the finalized
  \textsc{Close} call at which the coordinated attack can begin.
  Compute $\beta^{\mathsf{ano}}(\mathsf{ctx}_h,A)$ and publish
  signatures on both selected blocks under $\mathsf{pk}_i$.

  \item \textbf{Withdraw and recover collateral.}
  Request withdrawal in every finalized blockchain history
  in which $i$'s stake remains bonded, including both selected branches.
  Continue consensus participation to finalize the selected pair
  and advance each history until withdrawal takes effect,
  while maintaining transaction inclusion and the prescribed censorship.
  Once the attack succeeds and the settlement conditions hold,
  sign the refund authorization under both $a_i$ and $\mathsf{pk}_i$ and submit
  $\textsc{Settle}(a_i,i,\mathsf{cert})$ on each branch containing the deposit.
  Authorize only this account's refund;
  never transfer the account or authorize another player's refund.
\end{enumerate}
While waiting for data, a context, or signatures,
follow $\Pi_{\mathsf{bc}}$ subject to the prescribed censorship.
Missing or invalid inputs prevent the corresponding next step,
and waiting may persist through $T$.

\paragraph{Timing.}
The prescribed-play clause of \cref{ass:timely-execution} applies to
$s^{\mathsf{ano}}$.
It covers registration, \textsc{Close}, withdrawal requests, and eligible
\textsc{Settle} calls on every blockchain branch,
together with the blockchain progress needed to complete these steps by $T$.
In addition, whenever coordination activates under prescribed play,
all rational validators use the same feasible context $\mathsf{ctx}_h$
to compute $\beta^{\mathsf{ano}}(\mathsf{ctx}_h,A)$
and receive every rational validator's signature pair on the selected blocks
in time for the prescribed withdrawal requests to execute before the
initial episode boundary.

\paragraph{Unlinkability.}
We assume \emph{unlinkability} for prescribed enrollment as follows.
Fix a rational validator $i$ following the prescribed enrollment strategy.
Consider an alternative execution in which $i$ is honest and owns no
enrollment accounts, while other validators fund and operate the same
account clients without access to $\mathsf{sk}_i$.
Before activation, the joint distribution of public observations,
validator inputs, and message delivery is identical in both executions.

\paragraph{Utility.}
Each validator and its user accounts form one strategic player with a shared utility.
Adapt $\Attack_i$ from \cref{sec:coordination-game} to the selected pair
$\beta^{\mathsf{ano}}(\mathsf{ctx}_h,A)$,
with an account of $i$ in the activated set $A$ in place of
validator enrollment.
Let $d_i(\omega)$ count $i$'s distinct accounts benefiting from at least
one successful selected attack, each counted once.
Thus $d_i=0$ if $\Attack_i$ is false;
otherwise it counts all of $i$'s accounts in the common $A$.
Retain $\mathsf{Slashed}_i$ and $\mathsf{ExitFailed}_i$ from
\cref{sec:coordination-game}, requiring full-stake withdrawal on both
branches of every successful selected attack benefiting an account of $i$.

Extend the collateral valuation of \cref{sec:evidence-rewards} to
multiple deposits:
let $L_i(\omega)$ be the largest total amount of $i$'s collateral
remaining unreturned on any terminal finalized blockchain branch at time $T$.
Only accepted deposits in this contract instance count.
Locked collateral counts as unreturned,
and $L_i(\omega)=0$ if none of $i$'s deposits finalized.
Define
\begin{equation}\label{eq:anonymous-collateral-utility}
  u_i^D(\omega)
  =\epsilon d_i(\omega)
  -s_i\ind\{\mathsf{Slashed}_i(\omega)
    \lor\mathsf{ExitFailed}_i(\omega)\}
  -L_i(\omega).
\end{equation}
Write $U_i^D(s;t)=\mathbb E_{\omega\mid s,t}[u_i^D(\omega)]$.
Returning collateral adds no gain.

\subsubsection{Equilibrium Analysis}

\begin{theorem}[Anonymous deposit equilibrium]
\label{thm:anonymous-deposit-equilibrium}
  Under \cref{ass:timely-execution}, suppose the slashing mechanism satisfies honest-stake safety,
  and the collateral bound satisfies
  \begin{equation}\label{eq:anonymous-deposit-collateral-bound}
    D\ge2\epsilon.
  \end{equation}
  Then, the strategy profile $s^{\mathsf{ano}}$ is an ex post Nash equilibrium.
  Moreover, for every type profile $t$ and rational validator $i\in R(t)$,
  \[
    U_i^D(s^{\mathsf{ano}};t)=
    \begin{cases}
      0, & |R(t)|<m,\\
      \epsilon, & |R(t)|\ge m.
    \end{cases}
  \]
  Every prescribed deposit is returned by $T$.
  When $R(t)\ne\varnothing$ and $|R(t)|\ge m$, the selected attack succeeds,
  no rational validator is slashed, and all required full-stake withdrawals
  complete by $T$.
\end{theorem}

\begin{proof}
  Fix $t$ and $i\in R(t)$.
  Under prescribed play, each rational validator enrolls one account
  and honest validators enroll none, so $|A|=|R(t)|$.

  Suppose $|R(t)|<m$.
  Coordination never activates,
  so every validator process follows $\Pi_{\mathsf{bc}}$ throughout $[0,T]$.
  By unlinkability, the observable execution has the same distribution
  if other validators operate the enrollment accounts
  and $i$ is honest and owns none.
  An executed slashing penalty against $i$ would therefore expose accepted
  evidence against an honest validator with the same probability,
  contrary to honest-stake safety.
  Thus $i$ does not get slashed.
  The contract aborts and returns every deposit by $T$;
  no attack gain or withdrawal requirement arises,
  so $U_i^D(s^{\mathsf{ano}};t)=0$.

  Under a unilateral deviation by $i$, the other rational validators enroll at
  most $|R(t)|-1$ accounts, and honest validators enroll none.
  Any attack gain requires activation and hence
  $d_i\ge m-|R(t)|+1\ge2$.
  Each validator key permits at most one refund per branch.
  Honest validators provide no refund authorizations,
  and other rational validators authorize only their own accounts' refunds.
  Hence at least $(d_i-1)D$ of $i$'s collateral remains unreturned
  on each branch extending the common activated outcome.
  Thus
  \begin{equation}\label{eq:anonymous-account-deviation-bound}
    u_i^D\le d_i\epsilon-(d_i-1)D.
  \end{equation}
  Since $D\ge2\epsilon$ and $d_i\ge2$,
  this utility is at most $(2-d_i)\epsilon\le0$.
  A deviation without an attack gain also gives utility at most zero,
  so no unilateral deviation improves the payoff.

  Suppose instead $|R(t)|\ge m$.
  Timely registration activates the contract with $|A|=|R(t)|$.
  By the immediate-observation convention of \cref{sec:coordination},
  all rational validators observe activation and begin the prescribed
  censorship before publishing attack signatures.
  Their total weight satisfies
  $W(R(t))=|R(t)|/n\ge q_{\mathrm{coord}}$,
  and they maintain censorship throughout the initial episode,
  including while collecting signatures.
  Before activation, the same unlinkability argument as above excludes slashing.
  Under prescribed play, exactly the rational validators publish valid
  signature pairs on the selected blocks.
  The timing assumptions give them the same context and selected pair,
  and each receives pairs from all $|R(t)|=|A|$ rational validators.
  Thus every rational validator selects the same signers
  and executes its component of the same fixed monopoly continuation.
  Their total weight meets $q_{\mathrm{coord}}$,
  so \cref{def:monopoly-thresholds} guarantees that this continuation
  finalizes the selected pair while maintaining censorship.
  As in the proof of \cref{thm:risk-free},
  timely withdrawal releases each rational validator's full stake on both selected
  branches and every other finalized history reaching the episode boundary.
  Histories remaining in the initial episode are protected by censorship;
  after withdrawal, the stake remains unbonded and cannot be slashed.
  Timely settlement returns every deposit on each branch containing it.
  Each rational validator benefits through one account.
  Since all these statements hold,
  $U_i^D(s^{\mathsf{ano}};t)=\epsilon$.

  The refund bound~\eqref{eq:anonymous-account-deviation-bound} applies
  equally to deviations in this case.
  With an attack gain, utility is at most $\epsilon$ when $d_i=1$
  and at most zero when $d_i\ge2$;
  without an attack gain, utility is at most zero.
  Hence no unilateral deviation improves the payoff.
  Taking expectations in both cases covers randomized deviations and
  establishes the ex post Nash inequalities of \cref{def:epne}
  for every type profile.
\end{proof}

\begin{corollary}[Strict incentive to join]
\label{cor:anonymous-join-incentive}
  Under the assumptions of \cref{thm:anonymous-deposit-equilibrium},
  fix a type profile $t$ with $|R(t)|\ge m$ and $i\in R(t)$.
  Let $h_i\in\mathcal S_i^{\rational}$ be the strategy that follows
  $\Pi_{\mathsf{bc}}$, creates no enrollment accounts,
  and never registers or deposits collateral.
  Then
  \[
    U_i^D(s^{\mathsf{ano}};t)
    =\epsilon
    >0
    =U_i^D\bigl((h_i,s_{-i}^{\mathsf{ano}});t\bigr).
  \]
  Thus honest behavior is not a best response to $s_{-i}^{\mathsf{ano}}$,
  with $i$'s type remaining rational.
\end{corollary}

\begin{proof}
  Under $h_i$, validator $i$ has no enrolled account or deposited collateral,
  so no attack gain, collateral loss, or withdrawal requirement arises.
  Since $i$ follows $\Pi_{\mathsf{bc}}$,
  honest-stake safety excludes accepted evidence against $i$
  available to the other validators or published by $i$.
  Hence no positive slashing penalty against $i$ executes in a finalized
  blockchain history, and $\mathsf{Slashed}_i$ is false.
  By \cref{eq:anonymous-collateral-utility},
  $U_i^D((h_i,s_{-i}^{\mathsf{ano}});t)=0$.
  Switching only $i$ to its prescribed coordination strategy yields
  $s^{\mathsf{ano}}$, which gives $i$ payoff $\epsilon>0$
  by \cref{thm:anonymous-deposit-equilibrium}.
\end{proof}

\bibliographystyle{alpha}
\bibliography{refs}




\end{document}